\documentclass[conference]{IEEEtran}


\usepackage{graphicx,colordvi,psfrag}
\usepackage{amsmath,amssymb}
\usepackage{epstopdf}
\usepackage[caption=false]{subfig}
\usepackage{epsfig,cite}
\usepackage{calc,pstricks, pgf, xcolor}
\usepackage{nicefrac}
\usepackage{enumerate}

\newtheorem{theorem}{Theorem}

\newtheorem{proposition}{Proposition}
\newtheorem{remark}{Remark}

\newtheorem{lemma}{Lemma}
\newenvironment{proof}[1][Proof]{\noindent\textbf{#1.} }{\ \rule{0.5em}{0.5em}}

\newcommand{\by}{\mathbf{y}}

\newcommand{\bx}{\mathbf{x}}

\newcommand{\bz}{\mathbf{z}}
\newcommand{\bZ}{\mathbf{Z}}

\newcommand{\bs}{\mathbf{s}}
\newcommand{\bS}{\mathbf{S}}

\newcommand{\RR}{\mathbb{R}}

\newcommand{\m}{\mathcal}
\newcommand{\simplex}{\m{P}^{|\m{Z}|}}
\newcommand{\bhs}{\mathbf{\hat{s}}^n}
\newcommand{\bhS}{\mathbf{\hat{S}}^n}
\newcommand\restr[2]{{% we make the whole thing an ordinary symbol
  \left.\kern-\nulldelimiterspace % automatically resize the bar with \right
  #1 % the function
  \vphantom{\big|} % pretend it's a little taller at normal size
  \right|_{#2} % this is the delimiter
  }}

\newcommand{\argmin}{\operatornamewithlimits{argmin}}

\newrgbcolor{BoxRed}{1 .5 .5}
\newrgbcolor{BoxBlue}{.9 .9 1}

\begin{document}
\allowdisplaybreaks

\title{Subset--Universal Lossy Compression}

\author{\IEEEauthorblockN{Or Ordentlich}
\IEEEauthorblockA{Tel Aviv University\\
ordent@eng.tau.ac.il}
\and
\IEEEauthorblockN{Ofer Shayevitz}
\IEEEauthorblockA{Tel Aviv University\\
ofersha@eng.tau.ac.il}
\IEEEoverridecommandlockouts
\IEEEcompsocitemizethanks{
\IEEEcompsocthanksitem
The work of O. Ordentlich was supported by the Adams Fellowship Program of the Israel Academy of Sciences and Humanities, a fellowship from The Yitzhak and Chaya Weinstein Research Institute for Signal Processing at Tel Aviv University and the Feder Family Award. The work of O. Shayevitz was supported by the Israel Science Foundation under grant agreement no. 1367/14.}}

\maketitle

\begin{abstract}
A lossy source code $\mathcal{C}$ with rate $R$ for a discrete memoryless source $S$ is called subset--universal if for every $0<R'< R$, almost every subset of $2^{nR'}$ of its codewords achieves average distortion close to the source's distortion-rate function $D(R')$. In this paper we prove the asymptotic existence of such codes. Moreover, we show the asymptotic existence of a code that is subset--universal with respect to all sources with the same alphabet. 
\end{abstract}

\section{Introduction}
\label{sec:Into}

To motivate the topic studied in this work, let us consider the following scenario. Let $S$ be a discrete source with a probability mass function (pmf) $P_S$ over an alphabet $\m{S}$. Assume one wants to convey $n$ i.i.d. instances of $S$ over a deterministic channel $f:\m{X}^n\mapsto\m{Y}^n$, with the smallest possible distortion. If the channel law $f(\cdot)$ is known at both ends, the channel corresponds to a bit pipe whose capacity is the normalized logarithm of the function's image size, i.e., $C=\tfrac{1}{n}\log{|f(\m{X}^n)|}$. In this case a separation approach, where the source is first compressed with rate $C$ and then the compression index is transmitted over the channel, is optimal and achieves the best possible average distortion $D_{P_S}(C)$, where $D_{P_S}(\cdot)$ is the distortion-rate function of the source $P_S$.

When only the encoder knows the function $f(\cdot)$ but the decoder does not, it may be possible to ``learn'' the channel law if the class of possible functions is sufficiently small. For instance, if $f(\cdot)$ tensorizes, i.e., if $f(\bx^n)=h(x_1),\ldots,h(x_n)$, it is possible to learn $f(\cdot)$ with negligible cost for $n$ large enough. In this case, the separation based approach continues to be asymptotically optimal. Moreover, sometimes even partial knowledge of $f(\cdot)$ at the decoder is sufficient for separation to be optimal. For example, if the ``channel'' is a binary memory device with $\approx pn$ cells stuck at either $0$ or $1$, and whose locations are known only to the encoder, a reliable communication rate approaching $\tfrac{1}{n}\log{|f(\m{X}^n)|}$ can be achieved via Gelfand-Pinsker coding~\cite{gp80,heg83}, and subsequently separation can asymptotically attain the optimal distortion.

A separation approach, however, is completely useless when $f(\cdot)$ can be any \emph{arbitrary} function known only to the encoder. This follows from the fact that the capacity of the compound channel $\m{F}$ that consists of all $|\m{Y}|^{n|\m{X}|^n}$ possible functions from $\m{X}^n$ to $\m{Y}^n$ is zero. It is therefore tempting to jump to the conclusion that the smallest average distortion that can be attained in this setting is $D_{P_S}(0)$. But is this indeed the case? We will show that the answer is negative, and that in fact there exists a coding scheme that achieves average distortion $D_{P_S}(\tfrac{1}{n}\log|f(\m{X}^n)|)$ for almost every deterministic channel.

Note that the encoder, which knows the function $f(\cdot)$, can deterministically impose any output from the set $f(\m{X}^n)\subseteq\m{Y}^n$. However, due to its ignorance, the decoder does not know how to translate the possible outputs to a vector of $\log|f(\m{X}^n)|$ bits. Instead, the encoder and decoder can agree in advance on a mapping $g:\m{Y}^n\mapsto \m{\hat{S}}^n$
from each possible channel output to a source reconstruction sequence. The set $g(\m{Y}^n)$ of all possible reconstructions can be thought of as a source code $\m{C}$ of rate $R=\log|\m{Y}|$ bits per source symbol. The channel's effect is in diluting $\m{C}$ to the source code $\m{C}'=g(f(\m{X}^n))\subseteq\m{C}$ whose rate is $R'=\tfrac{1}{n}\log|f(\m{X}^n)|$ bits per symbol. In other words, a deterministic channel $f(\cdot)$ with capacity $R'$ chooses a subset of $2^{nR'}$ codewords from the $2^{nR}$ codewords in $\m{C}$.

This motivates the study of \emph{subset--universal} source codes. A subset--universal source code of rate $R$ has the property that for any $0<R'< R$, almost every subset of $2^{nR'}$ of its codewords is close to being optimal, in the sense that its induced average distortion is close to the distortion-rate function $D_{P_S}(R')$. Our main result is the asymptotic existence of such codes. In fact, we show the asymptotic existence of a code that is subset--universal with respect to any i.i.d. source distribution over the alphabet $\m{S}$.

Returning to our discussion on joint-source channel coding over a deterministic channel known only at the encoder, the existence of a subset-universal code implies that for almost every choice of $f(\cdot)$, an average distortion of $D_{P_S}(\tfrac{1}{n}\log|f(\m{X}^n)|)$ can be asymptotically achieved.

\section{Main Result}
\label{sec:main}

The entropy of a random variable $Y$ with probability mass function (pmf) $P_Y$ over the alphabet $\m{Y}$, is defined as
\begin{align}
H(Y)\triangleq -\sum_{y\in\m{Y}}P_Y(y)\log P_Y(y).\nonumber
\end{align}
For a pair of random variables $(Y,Z)$ with a joint pmf $P_{YZ}=P_Y P_{Z|Y}$ over the product alphabet $\m{Y}\times\m{Z}$ , the conditional entropy is defined as
\begin{align}
H(Z|Y)\triangleq -\sum_{(y,z)\in(\m{Y}\times\m{Z})}P_{YZ}(y,z)\log P_{Z|Y}(z|y)\nonumber.
\end{align}
and the mutual information is defined as
\begin{align}
I(Y;Z)\triangleq H(Y)-H(Y|Z)=H(Z)-H(Z|Y).\nonumber
\end{align}
For two distributions $P$ and $Q$ defined on the same alphabet $\m{Z}$, the KL-divergence is defined as
\begin{align}
D(P||Q)\triangleq \sum_{z\in\m{Z}} P(z)\log\frac{P(z)}{Q(z)}.\nonumber
\end{align}

Let $S$ be a discrete memoryless source (DMS) over the alphabet $\m{S}$ with pmf $P_S$, $\m{\hat{S}}$ a reconstruction alphabet, and \mbox{$d:\m{S}\times\m{\hat{S}}\mapsto \RR^+$} a bounded single-letter distortion measure. A $(2^{nR},n)$ lossy source code consists of an encoder that assigns an index $m(\bs^n)\in\{1,2,\ldots,2^{nR}\}$ to each $n$-dimensional vector $\bs^n\in\m{S}^n$, and a decoder that assigns an estimate $\mathbf{\hat{s}}^n(m)\in\m{\hat{S}}^n$ to each index $m\in\{1,2,\ldots,2^{nR}\}$. The set $\m{C}=\left\{\mathbf{\hat{s}}^n(1),\ldots,\mathbf{\hat{s}}^n(2^{nR})\right\}$ constitutes the codebook. We will assume throughout that the encoder assigns to each source sequence the index $m^*(\bs^n)$ according to the optimal rule
\begin{align}
m^*(\bs^n)=\argmin_{m\in\{1,2,\ldots,2^{nR}\}}\frac{1}{n}\sum_{i=1}^n d\left(s_i,(\bhs(m))_i\right),\nonumber
\end{align}
such that the lossy source code is completely specified by the codebook $\m{C}$ (and the distortion measure $d$).

The expected distortion associated with a lossy source code $\m{C}$ is defined as
\begin{align}
\mathbb{E}_{\bS^n}\left(d(\bS^n,\mathbf{\hat{S}}^n)\right)\triangleq \mathbb{E}_{\bS^n}\left(\frac{1}{n}\sum_{i=1}^n d(S_i,\hat{S}_i) \right).\nonumber
\end{align}
A distortion-rate pair $(D,R)$ is said to be achievable if there exists a sequence of $(2^{nR},n)$ codes with $\limsup_{n\rightarrow \infty}\mathbb{E}_{\bS^n}\left(d(\bS^n,\mathbf{\hat{S}}^n)\right)\leq D$. The distortion-rate function $D_{P_S}(R)$ of the source $P_S$ is the infimum of distortions $D$ such that $(D,R)$ is achievable, and is given by~\cite{berger71,bg98}
\begin{align}
D_{P_S}(R)\triangleq \min_{P_{\hat{S}|S}:I(S;\hat{S})\leq R}\sum_{s\in\m{S},\hat{s}\in\m{\hat{S}}}P_S(s)P_{\hat{S}|S}(\hat{s}|s) d(s,\hat{s}).\label{dr}
\end{align}
For a sequence of $(2^{nR},n)$ codes and $0<R'<R$, we say that almost every subset with cardinality $2^{nR'}$ satisfies a certain property if the fraction of subsets with cardinality $2^{nR'}$ that do not satisfy the property vanishes with $n$.

In~\cite{ziv72}, Ziv proved that there exists a codebook with rate $R$ that asymptotically achieves the distortion-rate function $D(R)$, regardless of the underlying source distribution. His result holds for any stationary source and even for a certain class of nonstationary sources. Our main result, stated below, only deals with i.i.d. sources with unknown distribution, and is hence less general than~\cite{ziv72} in terms of the assumptions made on the source's distribution. However, it extends~\cite{ziv72} in the sense that the distortion attained by subsets, and not just the full codebook, is shown to be universally optimal.

\begin{theorem}
Let $S$ be a DMS over the alphabet $\m{S}$, $\mathcal{\hat{S}}$ a reconstruction alphabet, and $d:\m{S}\times \m{\hat{S}}\mapsto\RR^+$ a bounded distortion measure. For any $R>0$, $\delta>0$, there exists a sequence of $(2^{nR},n)$ codebooks $\mathcal{C}\subset\mathcal{\hat{S}}^n$ with the property that for any source pmf $P_S$ on $\m{S}$ and any $0<R'<R$, almost every subset of $2^{nR'}$ codewords from $\mathcal{C}$ achieves an expected distortion no greater than $D_{P_S}(R'-\delta)$.
\label{thm:universal}
\end{theorem}

\begin{remark}
Ziv's result~\cite{ziv72} is based on first splitting the source sequence into $\ell$ blocks of length $k$ ($\ell\gg k$). Then, the best $(2^{kR},k)$ lossy source code, with respect to the average distortion among the $\ell$ blocks, is found, and conveyed to the receiver. Afterwards, each block is encoded using this source code and its index is sent to the receiver. If the number of blocks of length $k$ is large enough, the overhead for transmitting the codebook becomes negligible, which allows the scheme to perform well for a large class of distributions. This construction, however, is not subset universal. To see this, note that the empirical distribution of most codewords in the obtained codebook is that of the largest type. Thus, a random subset of $2^{nR'}$ codewords will contain a negligible fraction of codewords from other types, and in general cannot attain the distortion-rate function unless the optimal reconstruction distribution at rate $R'$ is equal to that of the largest type as well.
\end{remark}

\vspace{1mm}

The proof of Theorem~\ref{thm:universal} relies on independently drawing the codewords of $\m{C}$ from a mixture distribution, and will be given in Section~\ref{sec:proof}. In the next section we lay the ground by proving certain properties of mixture distributions.

\section{Properties of Mixture Distributions}
\label{subsec:mixture}

We follow the notation of~\cite{elgamalkim}, and define the empirical pmf of an $n$-dimensional sequence $\by^n$ with elements from $\m{Y}$ as
\begin{align}
\pi(y|\by^n)\triangleq \tfrac{1}{n}\left|i \ : \ y_i=y \right| \ \text{for } y\in\m{Y}.\nonumber
\end{align}
Similarly, the empirical pmf of a pair of $n$-dimensional sequences $(\by^n,\bz^n)$ with elements from $\m{Y}\times\m{Z}$ is defined as
\begin{align}
\pi(y,z|\by^n,\bz^n)\triangleq \tfrac{1}{n}\left|i \hspace{0.2mm} : \hspace{0.2mm} (y_i,z_i)=(y,z) \right| \ \text{for } (y,z)\in\m{Y}\times\m{Z}.\nonumber
\end{align}
Let $P_{YZ}=P_Y P_{Z|Y}$ be a joint pmf on $\m{Y}\times{Z}$. The set of $\varepsilon$-typical $n$-dimensional sequences w.r.t. $P_Y$ is defined as
\begin{align}
\m{T}_{\varepsilon}^{(n)}(P_Y)\triangleq\left\{\by^n \ : \ |\pi(y|\by^n)-P_Y(y)|\leq \varepsilon P_Y(y) \ \forall y\in\m{Y} \right\},\nonumber
\end{align}
and the set of jointly $\varepsilon$-typical $n$-dimensional sequences w.r.t. $P_{YZ}$ is defined as
\begin{align}
&\m{T}_{\varepsilon}^{(n)}(P_{YZ})\triangleq\bigg\{(\by^n,\bz^n) \ : \nonumber\\ &\ |\pi(y,z|\by^n,\bz^n)-P_{YZ}(y,z)|\leq \varepsilon P_{YZ}(y,z) \ \forall (y,z)\in\m{Y}\times\m{Z} \bigg\}.\nonumber
\end{align}
We also define the set of conditionally $\varepsilon$-typical $n$-dimensional sequences w.r.t. $P_{YZ}$ as
\begin{align}
\m{T}_{\varepsilon}^{(n)}(P_{YZ}|\by^n)\triangleq\left\{\bz^n \ : \ (\by^n,\bz^n)\in\m{T}_{\varepsilon}^{(n)}(P_{YZ})  \right\}.\nonumber
\end{align}
The next statement follows from the definitions above~\cite{elgamalkim}.

\begin{proposition}
Let $\by^n\in\m{Y}^n$. For every $\bz^n\in\m{T}_{\varepsilon}^{(n)}(P_{YZ}|\by^n)$ we have $\bz^n\in\m{T}_{\varepsilon}^{(n)}(P_Z)$. If in addition, $\by^n\in\m{T}_{\varepsilon'}^{(n)}(P_Y)$, for some $\varepsilon'<\varepsilon$, then for $n$ large enough
\begin{align}
|\m{T}_{\varepsilon}^{(n)}(P_{YZ}|\by^n)|\geq (1-\varepsilon)2^{n(1-\varepsilon)H(Z|Y)}.\nonumber
\end{align}
\label{prop:typicality}
\end{proposition}

Let $\simplex$ denote the simplex containing all probability mass functions on $\m{Z}$. For every $\theta\in\simplex$, let $P_\theta(z)$ be the corresponding pmf evaluated at $z$. Let $w(\theta)$ be some probability density function on $\simplex$. We may now define the \emph{mixture} distribution $Q$ as~\cite{mf98}
\begin{align}
Q(\bz^n)=\int_{\theta\in\simplex}w(\theta)\prod_{i=1}^n P_{\theta}(z_i) d\theta.\label{mixpmf}
\end{align}
The following propositions are proved in the appendix.

\begin{proposition}
Let $\bZ^n$ be a random $n$-dimensional sequence drawn according to $Q(\bz^n)$ defined in~\eqref{mixpmf}. Let $P_{YZ}$ be some pmf on $\m{Y}\times\m{Z}$, and let $\by^n\in\m{T}_{\varepsilon'}^{(n)}(P_Y)$, for some $\varepsilon'<\varepsilon$. For $n$ large enough, we have
\begin{align}
&\Pr\bigg(\bZ^n\in\m{T}_{\varepsilon}^{(n)}(P_{YZ}|\by^n)\bigg)\nonumber\\
&\geq(1-\varepsilon)\int_{\theta\in\simplex}\hspace{-6mm}w(\theta) 2^{-n(1+\varepsilon)\left(I(Y;Z)+D(P_{Z}||P_\theta)+2\varepsilon H(Z|Y)\right)}d\theta.\nonumber
\end{align}
\label{prop:mixturetypicality}
\end{proposition}

\begin{proposition}
Let $P_Z$ be some distribution in $\simplex$. For any $0<\xi<1/|\m{Z}|^2$, there exists a subset $\m{V}\subset\simplex$ with Lebesgue measure $\xi^{|\m{Z}|-1}$ on $\simplex$, such that for all $P_{\theta}\in\m{V}$
\begin{align}
D(P_Z||P_\theta)\leq\log\frac{1}{1-\xi|\m{Z}|^2},\nonumber
\end{align}
\label{prop:dneighnorhood}
\end{proposition}

Combining Proposition~\ref{prop:mixturetypicality} and~\ref{prop:dneighnorhood} yields the following lemma.

\begin{lemma}
Let $Q(\bz^n)$ be as defined in~\eqref{mixpmf}, with $w(\theta)$ taken as the uniform distribution on $\simplex$, and let $\bZ^n$ be a random $n$-dimensional sequence drawn according to $Q(\bz^n)$. Let $P_{YZ}$ be some pmf on $\m{Y}\times\m{Z}$, and let $\by^n\in\m{T}_{\varepsilon'}^{(n)}(P_Y)$, for some $\varepsilon'<\varepsilon$. For $n$ large enough, we have
\begin{align}
&\Pr\bigg(\bZ^n\in\m{T}_{\varepsilon}^{(n)}(P_{YZ}|\by^n)\bigg)\geq 2^{-n\left(I(Y;Z)+\delta(\varepsilon)\right)},\nonumber
\end{align}
where $\delta(\epsilon)\rightarrow 0$ for $\epsilon\rightarrow 0$.
\label{lem:MIuniversalbound}
\end{lemma}

\begin{proof}
Let $c_{|\m{Z}|}$ be the Lebesgue measure of the simplex $\simplex$. Clearly, we have that $w(\theta) = c_{|\m{Z}|}^{-1}$ for $\theta\in\simplex$ and $w(\theta)=0$ otherwise.
Setting $\xi=(1-2^{-\varepsilon})/|\m{Z}|^2$ in Proposition~\ref{prop:dneighnorhood} implies that there exists a set $\m{V}\subset\simplex$ with volume $((1-2^{-\varepsilon})/|\m{Z}|^2)^{|\m{Z}|-1}$ such that $D(P_Z||P_\theta)<\varepsilon$ for any $P_{\theta}\in\m{V}$. Combining this with Proposition~\ref{prop:mixturetypicality} gives
\begin{align}
&\Pr\bigg(\bZ^n\in\m{T}_{\varepsilon}^{(n)}(P_{YZ}|\by^n)\bigg)\nonumber\\
&\geq(1-\varepsilon)\int_{\theta\in\m{V}}w(\theta) 2^{-n(1+\varepsilon)\left(I(Y;Z)+\varepsilon+2\varepsilon H(Z|Y)\right)}d\theta\nonumber\\
&>\frac{1-\varepsilon}{c_{|\m{Z}|}}\left(\frac{1-2^{-\varepsilon}}{|\m{Z}|^2}\right)^{|\m{Z}|-1}2^{-n(1+\varepsilon)\left(I(Y;Z)+\varepsilon+2\varepsilon H(Z|Y)\right)}\nonumber\\
&=2^{-n\left(I(Y;Z)+\delta(\epsilon)\right)},\nonumber
\end{align}
where $\delta(\epsilon)\rightarrow 0$ for $\epsilon\rightarrow 0$.
\end{proof}

\section{Proof of Theorem~\ref{thm:universal}}
\label{sec:proof}

\underline{\textit{Random codebook generation:}} Let
\begin{align}
Q(\bhs)=\int_{\theta\in\m{P}^{\m{|\hat{S}|}}} w(\theta)\prod_{i=1}^n P_{\theta}(\hat{s}_i)d\theta,\nonumber
\end{align}
where $\m{P}^{|\m{\hat{S}|}}$ is the simplex containing all probability mass functions on $\m{\hat{S}}$  and $w(\theta)$ is the uniform distribution on $\m{P}^{|\m{\hat{S}|}}$. Randomly and independently generate $2^{nR}$ sequences $\bhs(m), \ m\in\{1,2,\ldots,2^{nR}\}$, each according to $Q(\bhs)$. These sequence constitutes the full codebook $\m{C}$.

A subset of the codebook, indexed by $(\m{I},R')$, consists of an index set $\m{I}\subset\{1,2,\ldots,2^{nR}\}$ with cardinality $|\m{I}|=2^{nR'}$, $0<R'<R$, and the corresponding sequences $\bhs(m), \ m\in\m{I}$. An arbitrary $(\m{I},R')$ subset of $\m{C}$ is revealed to the encoder and the decoder.

\underline{\textit{Encoding:}} Given the source sequence $\bs^n$ and a subset $(\m{I},R')$ of $\m{C}$, the optimal encoder sends the index of the codeword that achieves the minimal distortion, i.e., it sends
\begin{align}
m^*=\argmin_{m\in\m{I}}\frac{1}{n}\sum_{i=1}^n d\left(s_i,(\bhs(m))_i\right).\label{eq:optenc}
\end{align}
Note that the optimal encoder~\eqref{eq:optenc} is universal, i.e., it does not have to know the true underlying source pmf, and that such knowledge can in no way improve its performance.

\underline{\textit{Decoding:}} Upon receiving the index $m$, the decoder simply sets the reconstruction sequence as $\bhs(m)$.

\underline{\textit{Analysis:}} Let us define a quantized grid of rates in $[0,R)$ whose resolution is $\Delta>0$
\begin{align}
\m{R}_{\Delta}\triangleq\left\{R_j \ : \ R_j=j\cdot \Delta, \ j=1,\ldots,\lfloor R/\Delta\rfloor  \right\}.\nonumber
\end{align}
We further define a set of quantized distributions as follows. For $0<\Delta<1$ and $0<p_{\text{min}}<1$ define the set of points
\begin{align}
\m{G}_{\Delta,p_{\text{min}},|\m{S}|}\triangleq p_{\text{min}}\cdot\Bigg\{&1,\left(1+\frac{\Delta}{|\m{S}|}\right),\left(1+\frac{\Delta}{|\m{S}|}\right)^2,\nonumber\\
& \ \ \ \ldots,\left(1+\frac{\Delta}{|\m{S}|}\right)^{\left\lfloor\frac{-\log{p_{\text{min}}}}{\log\left(1+\frac{\Delta}{|\m{S}|}\right)} \right\rfloor} \Bigg\}.\nonumber
\end{align}
Our set of quantized distributions is denoted by $\m{P}^{|\m{S}|}_{\Delta,p_{\text{min}}}$ and consists of all vectors in the simplex $\m{P}^{|\m{S}|}$ with at least $|\m{S}|-1$ components that belong to $\m{G}_{\Delta,p_{\text{min}},|\m{S}|}$. Clearly, the cardinality $|\m{P}^{|\m{S}|}_{\Delta,p_{\text{min}}}|$ of this set is a function of only $|\m{S}|$, $p_{\text{min}}$ and $\Delta$. Moreover, for any pmf $P_\theta\in\m{P}^{|\m{S}|}$ with the property that $\min_{s\in\m{S}}P_\theta(s)>(1-\Delta/|\m{S}|)p_{\text{min}}$, there exists a pmf $P_{\theta'}\in\m{P}^{|\m{S}|}_{\Delta,p_{\text{min}}}$ such that $|P_\theta(s)-P_{\theta'}(s)|\leq\Delta P_{\theta'}(s)$ for all $s\in\m{S}$. To see this, construct $P_{\theta'}$ by quantizing all but the greatest entry of $P_{\theta}$ to the nearest point in $\m{G}_{\Delta,p_{\text{min}},|\m{S}|}$ and set the remaining entry such that $\sum_{s\in\m{S}}P_{\theta'}(s)=1$.

Let $\m{P}^{|\m{S}|}_{p_{\text{min}}}$ be the set of all probability mass functions on $\m{S}$ whose minimal mass is smaller than $p_{\text{min}}$. We can partition the remaining part of the simplex as
\begin{align}
\m{P}^{|\m{S}|}\setminus\m{P}^{|\m{S}|}_{p_{\text{min}}}=\bigcup_{P_S\in\m{P}^{|\m{S}|}_{\Delta,p_{\text{min}}}}\m{P}(P_s)\nonumber
\end{align}
where the sets $\m{P}(P_s)\subset\m{P}^{|\m{S}|}$ are disjoint and satisfy the property that for any $P_\theta\in\m{P}(P_s)$ we have $|P_\theta(s)-P_S(s)|<\Delta P_S(s)$ for all $s\in\m{S}$.

Let $R'\in\m{R}_{\Delta}$ and $P_S\in\m{P}^{|\m{S}|}_{\Delta,p_{\text{min}}}$. We will next show that almost every codebook in our ensemble satisfies the property that almost every subset of $2^{nR'}$ of its codewords achieves average distortion no greater than $D_{P_S}(R'-\delta)$ with respect to all sources whose pmf belongs to $\m{P}(P_S)$. Since $|\m{R}_{\Delta}|\cdot |\m{P}^{|\m{S}|}_{\Delta,p_{\text{min}}}|$ does not increase with $n$, there must exist a codebook that simultaneously satisfies this property for all $R'\in\m{R}_{\Delta}$ and $P_S\in\m{P}^{|\m{S}|}_{\Delta,p_{\text{min}}}$. The theorem will then follow by taking $\Delta\rightarrow 0$, $p_{\text{min}}\rightarrow 0$, and using the continuity of $D_{P_S}(R)$ w.r.t. $R$ and $P_S$ for a bounded distortion measure $d$~\cite{csiszarkorner}.

Let $R'\in\m{R}_{\Delta}$, $P_S\in\m{P}^{|\m{S}|}_{\Delta,p_{\text{min}}}$ and $P_S'\in\m{P}(P_s)$. We upper bound the average distortion of a given subset  $(\m{I},R')$ of $\m{C}$ over the ensemble of codebooks and a DMS $\bS^n$ with pdf $P_S'$. To that end, rather than analyzing the average distortion attained by the encoder~\eqref{eq:optenc}, we analyze the average distortion attained by the following suboptimal encoder. The suboptimal encoder first solves the minimization problem
\begin{align}
P^{R'}_{\hat{S}|S}=\argmin_{P_{\hat{S}|S}:I(S;\hat{S})\leq R'-\delta}\sum_{s\in\m{S},\hat{s}\in\m{\hat{S}}}P_S(s)P_{\hat{S}|S}(\hat{s}|s) d(s,\hat{s}),\nonumber
\end{align}
and sets $P^{R'}_{S\hat{S}}=P_S P^{R'}_{\hat{S}|S}$ as the target joint pmf. Then, given a source sequence $\bs^n$, it looks for the smallest index $m\in\m{I}$ such that $(\bs^n,\bhs(m))\in\m{T}_{\varepsilon}^{(n)}(P^{R'}_{S \hat{S}})$ and sends it to the decoder. If no such index is found, the smallest $m\in\m{I}$ is sent to the decoder. Note that this suboptimal encoder is not universal, as it requires knowledge of $P_S$. Nevertheless, its average distortion is by definition greater than that of the optimal encoder which is universal, and will therefore indeed serve as an upper bound on the distortion of the universal encoder.

Let $\bS^n$ be a random source drawn i.i.d. according to $P'_S$ and define $\mathbf{\hat{S}}^n$ as the reconstruction obtained by applying the suboptimal encoder (and the decoder described above). Define the error event $\m{E}(\m{I},R')$ as
\begin{align}
\m{E}(\m{I},R')\triangleq\left\{(\bS^n,\bhs(m))\notin\m{T}_{\varepsilon}^{(n)}(P^{R'}_{S \hat{S}}) \ \text{for all } m\in\m{I} \right\}.
\end{align}
By the definition of $P^{R'}_{S \hat{S}}$ and the definition of the typical set, we have that if $\m{E}(\m{I},R')$ did not occur
\begin{align}
d(\bS^n,\mathbf{\hat{S}}^n)&\leq (1+\varepsilon)D_{P_S}(R'-\delta),\label{avgdist}
\end{align}
where $D_{P_S}(\cdot)$ is the distortion-rate function~\eqref{dr} with respect to $P_S$.
%
%Let $d_{\text{max}}\triangleq\max_{s\in\m{S},\hat{s}\in\m{\hat{S}}}d(s,\hat{s})$. The average distortion achieved by this encoder (and the corresponding decoder) is bounded as
%\begin{align}
%\mathbb{E}_{\bS^n}d(\bS^n,\mathbf{\hat{S}}^n)&\leq P\left(\bar{\m{E}}\right)(1+\varepsilon)\sum_{s\in\m{S},\hat{s}\in\m{\hat{S}}}P^{R'}_{S \hat{S}}(s,\hat{s}) d(s,\hat{s})\nonumber\\
%&+P\left({\m{E}}\right)d_{\text{max}}\label{typdist}\\
%&\leq P\left({\m{E}}\right)d_{\text{max}}+(1+\varepsilon)D(R'-\delta),\label{avgdist}
%\end{align}
%where~\eqref{typdist} follows from the definition of the typical set $\m{T}_{\varepsilon}^{(n)}(P^{R'}_{S \hat{S}})$, and $D(\cdot)$ is the distortion-rate function~\eqref{dr}.
Further, the error event satisfies $\m{E}(\m{I},R')\subset \m{E}_1\cup\m{E}_2(\m{I},R')$, where
\begin{align}
\m{E}_1&=\left\{\bS^n\notin\m{T}_{\varepsilon'}^{(n)}(P_S) \right\}\nonumber\\
\m{E}_2(\m{I},R')&=\bigg\{\bS^n\in\m{T}_{\varepsilon'}^{(n)}(P_S),\nonumber\\
& \bhs(m)\notin\m{T}_{\varepsilon}^{(n)}(P^{R'}_{S\hat{S}}|\bS^n) \ \text{for all } m\in\m{I} \bigg\},\label{E2j}
\end{align}
and $\Delta<\varepsilon'<\varepsilon$. Note that $\m{E}_1$ is independent of $\m{I}$ and that $\Pr(\m{E}_1)\rightarrow 0$ as $n\rightarrow \infty$ by the (weak) law of large numbers and the fact that $\Delta<\varepsilon'$. The second error event does depend on $\m{I}$. We will now show that its expected probability with respect to the ensemble of codebooks $\mathbb{E}_{\m{C}}\left(\Pr(\m{E}_2(\m{I},R'))\right)$ vanishes for almost all $\m{I}\subset\{1,2,\ldots,2^{nR}\}$ of cardinality $2^{nR'}$.

%
%Our goal is to show that there exists a sequence of codebooks $\m{C}$ with the property that for each $R_j\in\m{R}^{\Delta}$, almost every subset $(\m{I},R_j)$ achieves distortion close to $(1+\epsilon)D(R_j-\delta)$. %For any subset $(\m{I},R_j)$ let $P^{R_j}_{S\hat{S}}$ be the target joint pmf found by the encoder (which depends only on $R_j$ and not on $\m{I}$) and define the error event
%%\begin{align}
%%\m{E}_2(\m{I},R_j)&=\bigg\{\bS^n\in\m{T}_{\varepsilon}^{(n)}(P_S),\nonumber\\
%%& \bhs(m)\notin\m{T}_{\varepsilon}^{(n)}(P^{R_j}_{\hat{S}|S}|\bS^n) \ \text{for all } m\in\m{I} \bigg\}.\label{E2j}
%%\end{align}
%By the preceding discussion, it suffices to show that there exists a sequence of codebooks $\m{C}$ with the property that for any $R_j\in\m{R}^{\Delta}$ and almost every subset $(\m{I},R_j)$ the error probability $\Pr\left(\m{E}_2(\m{I},R_j)\right)$ can be made arbitrary small for $n$ large enough. We show that this is in fact true for almost all codebooks in our ensemble. This property will be needed in order to show that there exists one codebook which is simultaneously good for all source distributions.

Let $\m{U}$ be a random subset of indices drawn from the uniform distribution on all subsets of $\{1,\ldots,2^{nR}\}$ with cardinality $2^{nR'}$. % and define the error event $\m{E}_{\m{U}}\triangleq \m{E}_2(\m{U},R')$.
%Given the codebook $\m{C}$ and the index sets $\m{U}$, the randomness of the event $\m{E}_{\m{U}}$ is only w.r.t. the source realization $\bS^n$. We may bound the expectation of $\Pr(\m{E}_{\m{U}})$ w.r.t. the random codebook $\m{C}$ and the random index sets $\m{U}$ as
%\begin{align}\label{unionb}
%\mathbb{E}_{\m{C},\m{U}}\left(\Pr\left(\m{E}_{\m{U}}\right)\right)\leq\sum_{j=1}^{\lfloor R/\Delta\rfloor}\mathbb{E}_{\m{C},\m{I}_j}\left(\Pr\left(\m{E}_2(\m{I}_j,R_j)\right)\right).
%\end{align}
By the random symmetric generation process of the codewords in $\m{C}$, the value of $\mathbb{E}_{\m{C}|\m{U}}\left(\Pr\left(\m{E}_2(\m{U},R')\right)|\m{U}=\m{I}\right)$ depends on the index set $\m{I}$ only through its cardinality $2^{n R'}$. We therefore have
\begin{align}
&\mathbb{E}_{\m{C},\m{U}}\left(\Pr\left(\m{E}_2(\m{U},R')\right)\right)=\mathbb{E}_{\m{U}}\mathbb{E}_{\m{C}|\m{U}}\left(\Pr\left(\m{E}_2(\m{U},R')|\m{U}\right)\right)\nonumber\\
&=\mathbb{E}_{\m{U}}\bigg(\sum_{\bs^n\in\m{T}_{\varepsilon'}^{(n)}(P_S)}P'_{\bS^n}(\bs^n)\nonumber\\
&\Pr\left(\bhS(m)\notin\m{T}_{\varepsilon}^{(n)}(P^{R'}_{S\hat{S}}|\bs^n) \ \text{for all } m\in\m{U} \ | \ \m{U} \right)\bigg)\nonumber\\
&=\mathbb{E}_{\m{U}}\bigg(\sum_{\bs^n\in\m{T}_{\varepsilon'}^{(n)}(P_S)}P'_{\bS^n}(\bs^n)\nonumber\\
&\prod_{m\in\m{U}}\Pr\left(\bhS(m)\notin\m{T}_{\varepsilon}^{(n)}(P^{R'}_{S\hat{S}}|\bs^n) | \ \m{U} \right)\bigg)\nonumber\\
&=\sum_{\bs^n\in\m{T}_{\varepsilon'}^{(n)}(P_S)}P'_{\bS^n}(\bs^n)\left(\Pr\left(\bhS(1)\notin\m{T}_{\varepsilon}^{(n)}(P^{R'}_{S\hat{S}}|\bs^n)\right) \right)^{2^{nR'}}\label{Pnocodeword}
\end{align}
By Lemma~\ref{lem:MIuniversalbound} we have
\begin{align}
\Pr\left(\bhS(1)\in\m{T}_{\varepsilon}^{(n)}(P^{R'}_{S\hat{S}}|\bs^n)\right)\geq 2^{-n(I(S;\hat{S}^{R'})+\delta(\epsilon))}\nonumber,
\end{align}
where $I(S;\hat{S}^{R'})$ is the mutual information between $S$ and $\hat{S}$ under the target joint pmf $P^{R'}_{S\hat{S}}$. This implies that
\begin{align}
\bigg(\Pr\bigg(\bhS(1)&\notin\m{T}_{\varepsilon}^{(n)}(P^{R'}_{S\hat{S}}|\bs^n)\bigg)\bigg)^{2^{nR'}}\nonumber\\
&\leq \left(1-2^{-n(I(S;\hat{S}^{R'})+\delta(\epsilon))}\right)^{2^{nR'}}\nonumber\\
&\leq \exp\left\{-2^{n(R'-I(S;\hat{S}^{R'})-\delta(\epsilon))} \right\}\label{expineq}\\
&\leq \exp\left\{-2^{n(\delta-\delta(\epsilon))} \right\},\label{Rineq}
\end{align}
where~\eqref{expineq} follows from the inequality $(1-x)^k\leq e^{-kx}$ which holds for $x\in[0,1]$, and~\eqref{Rineq} follows from the definition of $P^{R'}_{\hat{S}|S}$. Combining~\eqref{Pnocodeword} and~\eqref{Rineq} we have
\begin{align}
\mathbb{E}_{\m{C},\m{U}}\left(\Pr\left(\m{E}_2(\m{U},R')\right)\right)\leq \exp\left\{-2^{n(\delta-\delta(\epsilon))} \right\}.\label{doubleexpt}
\end{align}
Thus, for any $\delta>0$ we may take $0<\Delta<\varepsilon'<\varepsilon$ sufficiently small such that $\delta(\varepsilon)<\delta$ and~\eqref{doubleexpt} can be made arbitrary small when increasing $n$. The proof is concluded by applying the following arguments:
\begin{itemize}
\item By Markov's inequality,~\eqref{doubleexpt} implies that for every $P'_S\in\m{P}(P_S)$ and almost every subset $\m{I}\subset\{1,2,\ldots,2^{nR}\}$ of cardinality $2^{nR'}$ the expectation $\mathbb{E}_{\m{C}}\left(\Pr\left(\m{E}_2(\m{I},R')\right)\right)$ vanishes for large $n$.
\item Applying Markov's inequality again, this time with respect to $\m{C}$ we see that almost all codebooks in the ensemble satisfy that for every $P'_S\in\m{P}(P_S)$ and almost every subset $\m{I}\subset\{1,2,\ldots,2^{nR}\}$ of cardinality $2^{nR'}$ the probability $\Pr\left(\m{E}_2(\m{I},R')\right)$ vanishes for large $n$.
\item As discussed above, $\Pr(\m{E}_1)\rightarrow 0$ as $n\rightarrow \infty$ for every $P'_S\in\m{P}(P_S)$, regardless of $\m{I}$ and $\m{C}$. By the union bound $\Pr\left(\m{E}(\m{I},R')\right)\leq \Pr(\m{E}_1)+\Pr\left(\m{E}_2(\m{I},R')\right)$. We therefore have that almost all codebooks in the ensemble satisfy that for every $P'_S\in\m{P}(P_S)$ and almost every subset $\m{I}\subset\{1,2,\ldots,2^{nR}\}$ of cardinality $2^{nR'}$ the probability $\Pr\left(\m{E}(\m{I},R')\right)$ vanishes for large $n$.
\item Since the distortion measure $d$ is bounded, this together with~\eqref{avgdist} implies that almost all codebooks in the ensemble satisfy that for every $P'_S\in\m{P}(P_S)$ and almost every subset $\m{I}\subset\{1,2,\ldots,2^{nR}\}$ of cardinality $2^{nR'}$ the average distortion approaches $(1+\varepsilon)D_{P_S}(R'-\delta)$ as $n$ increases.
\item Since $|\m{R}_{\Delta}|\cdot |\m{P}^{|\m{S}|}_{\Delta,p_{\text{min}}}|$ does not increase with $n$, there must exist a sequence of codebooks whose average distortion approaches $(1+\varepsilon)D_{P_S}(R'-\delta)$ for almost every subset $\m{I}\subset\{1,2,\ldots,2^{nR}\}$ of cardinality $2^{nR'}$, simultaneously for all $R'\in\m{R}_{\Delta}$, $P_S\in\m{P}^{|\m{S}|}_{\Delta,p_{\text{min}}}$ and $P'_S\in\m{P}(P_S)$.
\item The theorem follows by taking $\varepsilon\rightarrow 0$, $p_{\text{min}}\rightarrow 0$ and using the continuity of the function $D_{P_S}(R)$ with respect to $P_S$ and $R$ for a bounded distortion measure $d$~\cite{csiszarkorner}.
\end{itemize}

\begin{appendix}

\begin{proof}[Proof of Proposition~\ref{prop:mixturetypicality}]
From Proposition~\ref{prop:typicality} and the definition of $\varepsilon$-typical sequences, we have that $\pi(z|\bz^n)\leq (1+\epsilon)P_Z(z)$
for any $\bz^n\in\m{T}_{\varepsilon}^{(n)}(P_{YZ}|\by^n)$. We can therefore write
\begin{align}
&\Pr\left(\bZ^n\in\m{T}_{\varepsilon}^{(n)}(P_{YZ}|\by^n)\right)=\sum_{\bz^n\in\m{T}_{\varepsilon}^{(n)}(P_{YZ}|\by^n)}Q(\bz^n)\nonumber\\
&=\sum_{\bz^n\in\m{T}_{\varepsilon}^{(n)}(P_{YZ}|\by^n)}\int_{\theta\in\simplex}w(\theta)\prod_{i=1}^n P_{\theta}(z_i)d\theta\nonumber\\
&\geq\sum_{\bz^n\in\m{T}_{\varepsilon}^{(n)}(P_{YZ}|\by^n)}\int_{\theta\in\simplex}w(\theta)\prod_{z\in\m{Z}} P_{\theta}(z)^{n(1+\varepsilon)P_Z(z)}d\theta\nonumber\\
&\geq(1-\varepsilon)2^{n(1-\varepsilon)H(Z|Y)}\nonumber\\
& \ \ \cdot\int_{\theta\in\simplex}w(\theta) 2^{n(1+\epsilon)\sum_{z\in\m{Z}}P_Z(z)\log P_{\theta}(z)}d\theta\nonumber\\
&=(1-\varepsilon)\int_{\theta\in\simplex}w(\theta)2^{n E}d\theta,\nonumber
\end{align}
where the last inequality follows from Proposition~\ref{prop:typicality}, and
\begin{align}
E&=(1-\varepsilon)H(Z|Y)+(1+\varepsilon)\sum_{z\in\m{Z}}P_Z(z)\log P_{\theta}(z)\nonumber\\
&=-2\varepsilon H(Z|Y)+(1+\varepsilon)\left(H(Z|Y)-D(P_Z||P_\theta)-H(Z)\right)\nonumber\\
&\geq -(1+\varepsilon)\left(I(Y;Z)+D(P_Z||P_\theta)+2\varepsilon H(Z|Y)\right).\nonumber
\end{align}
as desired.
\end{proof}

\begin{proof}[Proof of Proposition~\ref{prop:dneighnorhood}]
Without loss of generality, we may assume $\m{Z}=\{1,\ldots,|\m{Z}|\}$ and that $P_Z(1)\leq P_Z(2)\leq\cdots\leq P_Z(|\m{Z}|)$.
Note that under these assumptions $(1/|\m{Z}|) \leq P_Z(|\m{Z}|)\leq 1$, and $P_Z(|\m{Z}|-1)\leq \frac{1}{2}$. Let us define the perturbation set
\begin{align}
\m{U}=\bigg\{(u_1,\ldots,u_{|\m{Z}|}) \ : & \ 0\leq u_i <\xi, \ \text{for } i=1,\ldots,|\m{Z}|-1\nonumber\\
&\ u_{|\m{Z}|}=-\sum_{i=1}^{|\m{Z}|-1}u_i \bigg\},\nonumber
\end{align}
and the set $\m{V}=P_Z+\m{U}$, where the sum is in the Minkowski sense. Clearly, $\m{V}\subset\simplex$ for any $0<\xi<1/|\m{Z}|^2$, and its Lebesgue measure on the simplex $\simplex$ is $|\m{V}|=\xi^{|\m{Z}|-1}$. In addition, for any $P_{\theta}\in\m{V}$ we have
\begin{align}
D(P_Z||P_\theta)&=\sum_{z\in\m{Z}}P_Z(z)\log\frac{P_Z(z)}{P_\theta(z)}\nonumber\\
&\leq P_Z(|\m{Z}|)\log\frac{P_Z(|\m{Z}|)}{P_Z(|\m{Z}|)-\xi(|\m{Z}|-1)}\nonumber\\
&\leq\log\frac{1/|\m{Z}|}{(1/|\m{Z}|)-\xi(|\m{Z}|-1)}\nonumber\\
&\leq\log\frac{1}{1-\xi|\m{Z}|^2}.\nonumber
\end{align}
\end{proof}

\end{appendix}

\bibliographystyle{IEEEtran}
\bibliography{SubsetUniversal_bib}

\end{document}